\documentclass[conference]{IEEEtran}
\IEEEoverridecommandlockouts

\usepackage{bm}
\usepackage{cite}
\usepackage{amsmath,amssymb,amsfonts,graphicx,amsthm,mathtools,nicefrac}
\usepackage{graphicx}
\usepackage{textcomp}
\usepackage{xcolor}
\usepackage{float}
\usepackage[caption=false,font=normalsize,labelfont=sf,textfont=sf]{subfig}
\usepackage[ruled,linesnumbered]{algorithm2e}
\usepackage{stfloats}
\usepackage{balance}

\def\BibTeX{{\rm B\kern-.05em{\sc i\kern-.025em b}\kern-.08em
    T\kern-.1667em\lower.7ex\hbox{E}\kern-.125emX}}

\allowdisplaybreaks

\newtheorem{theorem}{Theorem}[section]
\newtheorem{assumption}{Assumption}[section]
\newtheorem{remark}{Remark}[section]

\numberwithin{equation}{section}

\newtheorem{lemma}[theorem]{Lemma}

\DeclareMathOperator{\rank}{\mathrm{rank}}

\DeclareMathOperator{\trace}{trace}

\def\la {\left\langle}
\def\ra {\right\rangle} 
\def \lb{\left(}
\def \rb{\right)}

\newcommand{\matsnorm}[2]{\left\| #1\right\|_{{#2}}}

\newcommand{\fronorm}[1]{\ensuremath{\matsnorm{#1}{\footnotesize{\mathsf{F}}}}}
\newcommand{\opnorm}[1]{\ensuremath{\matsnorm{#1}{}}}

\newcommand{\twonorm}[1]{\ensuremath{\matsnorm{#1}{\footnotesize{2}}}}

\newcommand{\bfm}[1]{\bm{#1}}

\newcommand{\E}[2][]{\mathbb{E}_{#1} \left\{ #2 \rule{0mm}{3mm}\right\}}

\def\va{\bfm a}   \def\mA{\bfm A}  
\def\vb{\bfm b}   \def\mB{\bfm B}  
     \def\C{\mathbb{C}}
   \def\mD{\bfm D}  
\def\ve{\bfm e}     
    
\def\vg{\bfm g}   \def\mG{\bfm G}  
\def\vh{\bfm h}     
   \def\mI{\bfm I}

   \def\mU{\bfm U}  
   \def\mV{\bfm V}  
   \def\mW{\bfm W}  
\def\vx{\bfm x}   \def\mX{\bfm X}  
\def\vy{\bfm y}   \def\mY{\bfm Y}  
   \def\mZ{\bfm Z}

\def\calA{{\cal  A}}

\def\calD{{\cal  D}}

\def\calG{{\cal  G}} 
\def\calH{{\cal  H}} 
\def\calI{{\cal  I}}

\def\calM{{\cal  M}} 
 
\def\calO{{\cal  O}} 
\def\calP{{\cal  P}} 
 
\def\calR{{\cal  R}} 
\def\calS{{\cal  S}}

\newcommand{\bfsym}[1]{\bm{#1}}

             \def\bSigma{\bfsym \Sigma}

\def \calGT {\calG^\ast}

\def \tran {\mathsf{T}}
\def \tranH{\mathsf{H}}

\def \bzero{\bm 0}

\usepackage{bbm}

\def \tG {\widetilde{\calG}}

\begin{document}

\title{Riemannian Gradient Descent Method to Joint Blind Super-Resolution and Demixing in ISAC\thanks{Corresponding author: Jinchi Chen}
}

\author{
	\IEEEauthorblockN{
	Zeyu Xiang\IEEEauthorrefmark{1},
	Haifeng Wang\IEEEauthorrefmark{2},
	Jiayi Lv\IEEEauthorrefmark{1},
	Yujie Wang\IEEEauthorrefmark{1},
	Yuxue Wang\IEEEauthorrefmark{1},
	Yuxuan Ma\IEEEauthorrefmark{1},
	Jinchi Chen\IEEEauthorrefmark{1}}
\IEEEauthorblockA{\IEEEauthorrefmark{1}School of Mathematics, East China University of Science and Technology, Shanghai, China
	\\\{jcchen.phys\}@gmail.com}
\IEEEauthorblockA{\IEEEauthorrefmark{2}China Mobile (Zhejiang) Research \& Innovation Institute, Hangzhou, China
	\\\{wanghaifeng40\}@zj.chinamobile.com}
}

\maketitle

\begin{abstract}
Integrated Sensing and Communication (ISAC) has emerged as a promising technology for next-generation wireless networks. In this work, we tackle an ill-posed parameter estimation problem within ISAC, formulating it as a joint blind super-resolution and demixing problem. Leveraging the low-rank structures of the vectorized Hankel matrices associated with the unknown parameters, we propose a Riemannian gradient descent (RGD) method. Our theoretical analysis demonstrates that the proposed method achieves linear convergence to the target matrices under standard assumptions. Additionally, extensive numerical experiments validate the effectiveness of the proposed approach.

\end{abstract}


\section{Introduction}

%

ISAC represents an innovative and emerging paradigm aimed at enhancing spectral efficiency by enabling the joint utilization of spectrum resources for both communication and sensing functionalities \cite{lu2024integrated,gonzalez2024integrated,wen2024survey}. By integrating these two functions, ISAC seeks to reduce spectrum congestion and improve the overall performance of wireless systems.

In practical implementations of ISAC, several significant challenges arise. One major challenge is that the transmitted waveform in passive \cite{zheng2017super,sedighi2021localization} or multistatic \cite{dokhanchi2019mmwave} radar systems is often unknown to the receiver. This uncertainty complicates the ability of the receiver to effectively detect and process signals, leading to a decrease in system performance.  

Another critical challenge in ISAC systems is the acquisition of accurate channel state information (CSI) in communication systems, which is inherently complex. For example, in Terahertz (THz) band communication systems \cite{mishra2019toward,lu2024integrated}, the channel coherence time is extremely short, making it difficult to obtain reliable CSI using conventional pilot-based methods. The rapid variability of THz channels, combined with severe path loss and molecular absorption, further complicates the acquisition process. 


In this paper, we consider a generalized scenario where both the radar and communication channels, along with their transmitted signals, are unknown to the common receiver. We formulate the task of estimating channel parameters and transmitted signals as a Joint Blind Super-resolution and Demixing (JBSD) problem. Specifically, we use a subspace lifting technique to leverage the low-dimensional structures inherent in the data matrices related to the channel parameters and transmitted signals. This allows us to frame the JBSD problem as a low-rank matrix demixing problem. To solve the low-rank matrix demixing problem, we propose a novel Riemannian gradient descent framework. This approach enables computationally efficient reconstruction, with exact recovery guarantees established under standard assumptions.





\section{Problem Formulation}
Consider a scenario involving $K$ radar or communication users transmitting signals to a common receiver. The received signal at the receiver is modeled as a superposition of convolutions from $K$ users, expressed mathematically as:
\begin{align*}
	y(t) = \sum_{k=1}^{K} \sum_{p=1}^{r_k} d_{k,p}\delta(t-\tau_{k,p})\ast g_{k}(t),
\end{align*}
where $r_k$ represents  the number of paths for the $k$-th user, $\tau_{k,p}\in [0,1)$ denotes the delay, $d_{k,p}$ is the complex-valued amplitude, and $g_{k}(t)$ is the waveform corresponding to the $k$-th user. By taking the Fourier transform and sampling, we obtain for $j=0,1,\cdots, n$:
\begin{align}\label{eq: observation}
	y[j] = \sum_{k=1}^{K} \sum_{p = 1}^{r} d_{k,p} e^{-\imath 2\pi (j-1)\tau_{k,p}} \hat{g}_{k}[j],
\end{align}
where $\hat{g}_{k}$ represents the Fourier transform of $g_{k}(t)$. Defining $\vg_{k} = \begin{bmatrix}
	g_{k}[0] &\cdots & g_{k}[n-1]
\end{bmatrix}^\tran\in\C^n$, the objective of the JBSD problem is to estimate the parameters $\{d_{k,p}, \tau_{k,p}\}$ as well as the unknown signals $\{\vg_{k}\}$ from the observations in the equation above.

Since the number of measurements is less than the number of unknowns, JBSD is an ill-posed problem. To address this challenge, we adopt a subspace assumption inspired by prior research \cite{vargas2023dual,monsalve2022beurling,jacome2023multi,razavikia2023off}. Specifically, we assume that each signal $\vg_{k}$ lies in a low-dimensional subspace defined by a matrix $\mB_k\in\C^{n\times s_k}$, such that
\begin{align*}
	\vg_{k} = \mB_{k}\vh_{k}, \quad p=1, \cdots, r.
\end{align*}
where $\vh_k\in\C^{s_k}$ is an unknown coefficient vector. This assumption implies that the waveform $\vg_k$ can be represented as a linear combination of a redundant codebook matrix $\mB_k$. 

Under the subspace assumption, the measurements can be expressed as:
\begin{align}
	\label{eq sampling model}
	y[j] =\sum_{k=1}^{K}\la \vb_{k,j}\ve_j^\tran, \mX_{k,\natural}\ra,
\end{align}
where $\mX_{k,\natural} = \sum_{p= 1}^{r_k} d_{k,p} \vh_{k}\va_{\tau_{k,p}}^\tran$, the steering vector $\va_{\tau}\in\C^{n}$ is defined by:
\begin{align*}
 \begin{bmatrix}
	1 & e^{-\imath2\pi \tau} &\cdots & e^{-\imath2\pi\tau(n-1)}
\end{bmatrix}^\tran,
\end{align*}
and $\la\cdot, \cdot\ra$ represents the inner product, defined as $\la\mA,\mB\ra=\trace(\mA^\tranH\mB)$. Without loss of generality, we assume $r_1=\cdots=r_K=r$ and $s_1=\cdots=s_K=s$. 
Let $\calA_k: \C^{s\times n}\rightarrow\C^n$ be a linear operator defined by:
\begin{align*}
	\calA_k(\mX)[j] = \la \vb_{k,j}\ve_j^\tran, \mX\ra,
\end{align*}
with the adjoint operator $\calA^\ast_k$ given by $\calA^\ast_k(\vy)=\sum_{j=1}^{n}\vy[j]\vb_{j,k}\ve_j^\tran$.
Then the measurements in can be expressed as a more compact form:
\begin{align}
	\label{eq: sample}
	\vy =\sum_{k=1}^{K} \calA_k(\mX_{k,\natural}).
\end{align}
The problem of JBSD can thus be formulated as the task of demixing a sequence of matrices $\{\mX_k\}$ from the superposed linear measurements. Once the data matrices are recovered, the delays $\{\tau_{k,p}\}$ can be estimated using spatial smoothing MUSIC \cite{chen2022vectorized,evans1981high, evans1982application,yang2019source}, and the parameters $\{d_{k,p}, \vh_{k}\}$ can be recovered using an overparameterized linear system \cite{chen2024fast}. Thus the main objective of this work is to efficiently recover the data matrices $\{\mX_{k,\natural}\}$ from the measurement \eqref{eq sampling model}. 


When $K=1$, JBSD reduces to the classical blind super-resolution problem. Convex and non-convex approaches have been successfully employed to solve this problem. In particular, the theoretical guarantees of convex methods, including atomic norm minimization \cite{chi2016guaranteed,yang2016super, li2019atomic} and nuclear norm minimization \cite{chen2022vectorized}, have been established. However, these convex optimization-based methods are computationally expensive. To overcome this issue, non-convex optimization-based methods, such as low-rank factorization \cite{mao2022blind, li2024simpler} and low-rank manifold-based \cite{zhu2021blind} approaches, have gained attention for their efficiency in leveraging the low-rank structures of the data matrices.

Given that the measurements are a superposition of multiple users' data, JBSD presents a greater challenge compared to blind super-resolution. Consequently, theoretical guarantees and algorithms developed for blind super-resolution cannot be directly applied to JBSD. Recent research has focused on JBSD, with notable progress achieved through methods such as ANM (Atomic Norm Minimization) \cite{vargas2023dual,jacome2024multi,daei2024timely} and vectorized Hankel lifts\cite{monsalve2023beurling,wang2023blind}. However, compared to blind super-resolution, computationally efficient algorithms for JBSD remain limited. Therefore, developing fast and robust algorithms for JBSD is a critical area of interest. In this work, we propose a Riemannian gradient descent method for solving the aforementioned optimization problem efficiently. 

\section{Algorithms}
Let $\calH:\C^{s\times n}\rightarrow \C^{sn_1\times n_2}$ denote the vectorized Hankel lift operator defined as follows:
\begin{align*}
	\calH(\mX) = \begin{bmatrix}
		\vx_{0} & \vx_1 &\cdots & \vx_{n_2-1}\\
		\vx_{1} & \vx_2 &\cdots  & \vx_{n_2}\\
		\vdots &\vdots &\ddots  &\vdots\\
		\vx_{n_1-1} &\vx_{n_1} &\cdots & \vx_{n-1}\\
	\end{bmatrix}\in\C^{sn_1\times n_2},
\end{align*}
where $\vx_{j}$ is the $(j+1)$-th column of $\mX$ with $j=0,1,\cdots, n-1$, and $n_1+n_2=n+1$. It has shown that $\rank(\calH(\mX_{k,\natural}))=r$ \cite{chen2022vectorized}. Consequently, we consider the following non-convex optimization problem to recover the data matrices:
\begin{align}
	\label{tmp}
	\min_{\mX_k} ~&\frac{1}{2}  \twonorm{\mD\vy - \sum_{k=1}^{K} \calA_k\calD(\mX_k)}^2 ~\text{s.t.}~\rank(\calH(\mX_k)) = r.
\end{align}
Furthermore, we define the operator $\calD:\C^{s\times n}\rightarrow\C^{s\times n}$ as follows:
\begin{align*}
	\calD(\mX) = \begin{bmatrix}
		\sqrt{w_0} \vx_0 &\cdots & \sqrt{w_{n-1}} \vx_{n-1}
	\end{bmatrix},
\end{align*}
where $w_i= \#\{ j+k = i, 0\leq j\leq n_1-, 0\leq k\leq n_2-1 \} $. Let $\mD = \calD(\mI_{n})$ and $\calG = \calH\calD^{-1}$. Denoting $\mZ_k= \calG\calD(\mX_k)$, the optimization \eqref{tmp} can be reformulated as follows:
\begin{align*}
	\min_{\mZ_k}~&\frac{1}{2}\twonorm{\mD\vy - \sum_{k=1}^{K}\calA_k\calG^\ast(\mZ_k)}, \\
	\text{s.t.}~ & \mZ_k\in\calM_{k,r}, ~(\calI - \calG\calG^\ast)(\mZ_k) = \bzero, ~k=1,\cdots, K.
\end{align*}
where $\calM_{k,r}$ is the Riemannian manifold of all rank-$r$ complex matrix, embedded with inner product, the second constraint guarantees that $\mZ_k$ has the vectorized Hankel structure. Define $\mZ = (\mZ_1, \cdots, \mZ_K)$. Let $\calM_r= \calM_{1,r}\times \cdots \times \calM_{K,r}$  be the product manifold. We also consider the following optimization problem:
\begin{align*}
	\min_{\mZ\in\calM_r} f(\mZ),
\end{align*}
where 
\begin{align}
	\label{opti}
	f(\mZ) :&= \frac{1}{2} \twonorm{\sum_{k=1}^{K}\calA_k\calG^\ast(\mZ_k - \mZ_{k,\natural})}^2 \notag\\
	&\qquad+ \frac{1}{2}\sum_{k=1}^{K}\fronorm{(\calI - \calG\calG^\ast)(\mZ_k)}^2,		
\end{align}

We employ the Riemannian Gradient Descent (RGD) method \cite{boumal2023introduction} for the problem \eqref{opti}, which is summarized in Algorithm \ref{alg: SGD}. The RGD algorithm generates a sequence of iterates using the following update rule:
\begin{align*}
	\mZ_{k, t+1} = \calR_k\lb \mZ_{k,t} - \eta_t  \nabla_{\calM_{k,r}} f(\mZ_{k,t})\rb,
\end{align*}
where $\eta_t$ represents the step size, $ \nabla_{\calM_{k,r}} f(\mZ_{k,t})$ is the Riemannian gradient at $\mZ_{k, t}$, and $\calR_k(\cdot)$ is the retraction operator. By adopting the canonical Riemannian metric and employing truncated Singular Value Decomposition (SVD) as the retraction, the RGD method becomes:
\begin{align*}
	\mZ_{k,t+1} = \calP_r\lb \mZ_{k, t} - \alpha_t \calP_{T_{k,t}} (\mG_{k,t})  \rb,
\end{align*}
where $\mG_{k,t}$ is given by
\begin{align}
	\label{gradient}
	\mG_{k,t} = \sum_{\ell=1}^{K} \calG\calA_k^\ast \calA_\ell\calGT(\mZ_{\ell, t} - \mZ_{\ell,\natural}) + (\calI - \calG\calGT)(\mZ_{k, t}),
\end{align}
and $\calP_{T_{k,t}}$ represents the projection onto the tangent space of the manifold $\calM_{k,r}$ at $\mZ_{k, t}$. Let $\mZ_{k,t} = \mU_{k,t}\bSigma_{k,t}\mV_{k,t}^\tranH$ be the compact SVD. The tangent space $T_{k,t}$ of $\calM_{k,r}$ at $\mZ_{k, t}$ is defined as follows:
\begin{align*}
	T_{k,t} =\{\mU_{k,t}\mA_{k,t}^\tranH + \mB_{k,t}\mV_{k,t}^\tranH: \mA_{k,t}\in\C^{n_2\times r}, \mB_{k,t}\in\C^{sn_1\times r} \}.
\end{align*}
 The projection $\calP_{T_{k,t}}(\mY_k)$ is defined as follows:
\begin{align*}
	\calP_{T_{k,t}}(\mY_k) = \mU_{k,t}\mU_{k,t}^\tranH \mY_k + \mY_k\mV_{k,t}\mV_{k,t}^\tranH - \mU_{k,t}\mU_{k,t}^\tranH \mY_k \mV_{k,t}\mV_{k,t}^\tranH.
\end{align*}

Indeed, the RGD algorithm can be regarded as a generalization of Fast Iterative Hard Thresholding (FIHT) in \cite{cai2015fast,wei2016guarantees,zhu2021blind} to the JBSD problem. Moreover, the RGD method can be eﬃciently implemented, where the main computational complexity in each step is $\calO(K(r^2sn+r^3+srn\log n))$.

\begin{algorithm}[ht!]
	\caption{RGD--JBSD}
	\label{alg: SGD}
	
	\For{$t=0,1,\cdots, T-1$}{
		\tcp{Fully parallel}
		\For{$k=1,\cdots, K$}{
			Compute the gradient $\mG_{k,t}$ via \eqref{gradient}\;
			Update on the tangent space: 
			$\mW_{k,t} = \mZ_{k, t} - \alpha_t \calP_{T_{k,t}}(\mG_{k,t})$\;
			Retaction:
			$\mZ_{k,t+1} =\calP_{r}(\mW_{k,t})$\;
		}
	}
	
\end{algorithm}

\section{Theoretical Results}
In this section, we present our primary result based on the following two assumptions.
\begin{assumption}[$\mu_0$-incoherence]
	\label{assumption 0}
	Suppose that the columns $\{\vb_{k,i}\}$ of $\mB_k^\tranH$ for $k=1,\cdots, K$, are i.i.d sampled from the a distribution $F$, which satisfies the following conditions for $\vb\sim F$:
	\begin{align*}
		\E{\vb} = \bzero,
		\E{\vb\vb^\tranH}=\mI,
		\max_{1\leq p\leq s}|\vb[p]| \leq \sqrt{\mu_0}. 
	\end{align*}
\end{assumption}
\begin{assumption}[$\mu_1$-incoherence]
	\label{assumption 1}
	Let $\mZ_{k,\natural} = \mU_{k,\natural} \bSigma_{k,\natural}\mV_{k,\natural}^\tranH$ be the singular value decomposition of $\mZ_{k,\natural}$, where $\mU_{k,\natural}\in\C^{sn_1\times r}, \mV_{k,\natural}\in\C^{n_2\times r}$. Let $\mU_{k,\natural, j} = \mU_{k,\natural}[js:(j+1)s-1,:]\in\C^{s\times r}$ be the $j$-th block of $\mU_{k,\natural}$ for $j=0,\cdots, n_1-1$. Suppose that for all $k=1,\cdots, K$, the matrix $\mZ_{k,\natural} $ obeys the following conditions:
	\begin{align*}
		\max_{0\leq j\leq n_1-1} \fronorm{\mU_{k,\natural, j} }^2 \leq \frac{\mu_1r}{n}\text{ and } \max_{0\leq \ell\leq n_2-1} \twonorm{\ve_\ell^\tran\mV_{k,\natural}}^2\leq \frac{\mu_1r}{n}
	\end{align*}
	for some positive constant $\mu_1$.
\end{assumption}
We are now prepared to formally state our main result.
\begin{theorem}
	\label{them}
	Assume that Assumptions \ref{assumption 0} and \ref{assumption 1} hold. If the number of measurements satisfies $n\geq C_{\gamma} K^2s^2r^2 \kappa^2\mu_0^2\mu_1 \log^2(sn)$, then with probability at least $1-(sn)^{-\gamma}$,  the iterations produced by Algorithm \ref{alg: SGD} satisfy
\begin{align}
	\label{main result}
\sum_{k=1}^{K}\fronorm{\mZ_k - \mZ_{k,\natural}}^2 \leq \frac{1}{2^t} \cdot \frac{\sigma_0^2}{K\mu_0s(1+\varepsilon)}
\end{align}
for $t=0,1,\cdots, T$, where $\sigma_0^2 = \sum_{k=1}^{K}\sigma_r^2(\mZ_{k, \natural})$ and $\kappa =\frac{\max_k \sigma_1(\mZ_{k, \natural})}{\min_k \sigma_r(\mZ_{k, \natural})}$.

\end{theorem}
\begin{remark}
	Theorem \ref{them} establishes that RGD converges to the target matrices at a linear rate. Furthermore, the convergence rate is independent of the condition number of the target matrix, underscoring the efficiency of the proposed method.
\end{remark}

\section{Numerical Experiments}
In this section, we assess the performance of the proposed method and compare it to the Scaled Gradient Descent (Scaled--GD) method \cite{chen2024fast}. All numerical experiments were conducted using MATLAB R2022b on a macOS system equipped with a multi-core Intel CPU running at 2.3 GHz and 16 GB of RAM.

In our experiments, the data matrix $\mX_{k,\natural}$ is constructed as $\mX_{k,\natural} = \sum_{p=1}^{r}d_{k,p}\vh_k \va_{\tau_{k,p}}^\tran$. The amplitudes $\{d_{k,p}\}$ are generated in the form $(1+10^{c_{k,p}})e^{-\imath\varphi_{k,p}}$, where $c_{k,p}$ is uniformly sampled from $[0,1]$ and $\varphi_{k,p}$ is uniformly distributed over $[0,2\pi)$. The coefficient vector $\vh_k$ is a standard Gaussian random vector, subsequently normalized. For data matrices without frequency separation, 
the time delay parameters $\{\tau_{k,p}\}$ are uniformly sampled from $[0,1)$. For data matrices with frequency separation, $\{\tau_{k,p}\}_{p=1,\cdots, r}$ are sampled uniformly from $[0,1)$, ensuring the minimum separation satisfies $\min_{p_1\neq p_2}|\tau_{k,p_1}  - \tau_{k,p_2}| \geq \frac{1}{n}$. Additionally, the subspace matrices $\{\mB_{k}\}$ are i.i.d random matrices with entries uniformly sampled from $[-\sqrt{3}, \sqrt{3}]$.  We conduct $20$ Monte Carlo trials and consider the recovery successful if the relative error satisfies the condition
\begin{align*}
	\sqrt{\frac{\sum_{k=1}^{K}\fronorm{\hat{\mX}_k - \mX_{k,\natural}}^2}{ \sum_{k=1}^{K}\fronorm{ \mX_{k,\natural}}^2 } } \leq 10^{-3}.
\end{align*}
The algorithms are terminated when the relative error falls below $10^{-4}$ or the number of iterations exceeds $2000$.

In the first experiment, we examine the recovery performance of RGD in comparison to GD and Scaled--GD using the empirical phase transition framework. We set $n=160, s=K=2$ and vary $r$ in the range $\{2:1:8\}$. Fig \ref{fig: empir prob} presents the phase transition plots both with and without imposing the separation condition. The results indicate that RGD is more robust to the frequency separation condition and exhibits a higher phase transition threshold compared to the GD method.
\begin{figure}[htbp]
	\centering
	\subfloat[]{
		\centering
		\includegraphics[width=0.23\textwidth]{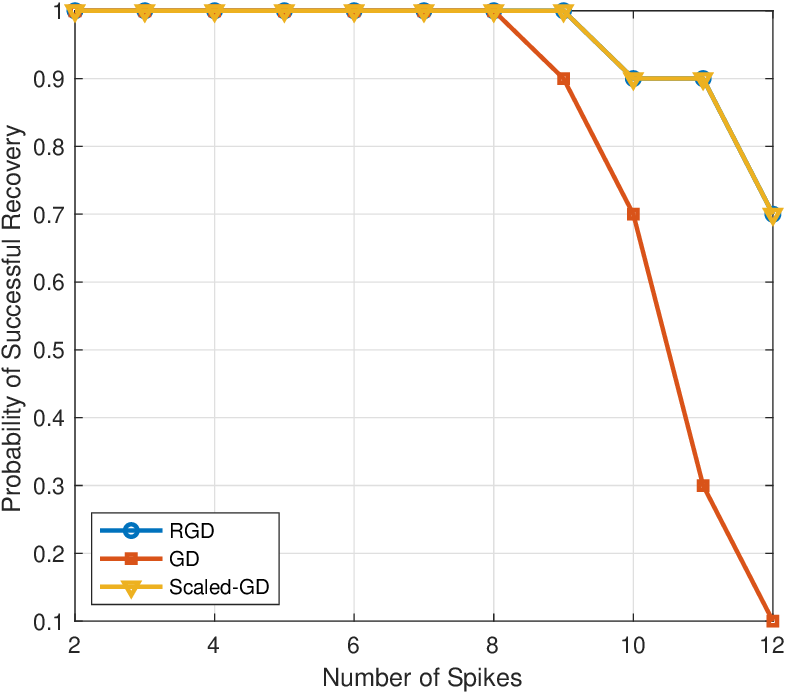}
	}
	\subfloat[]{
		\centering
		\includegraphics[width=0.23\textwidth]{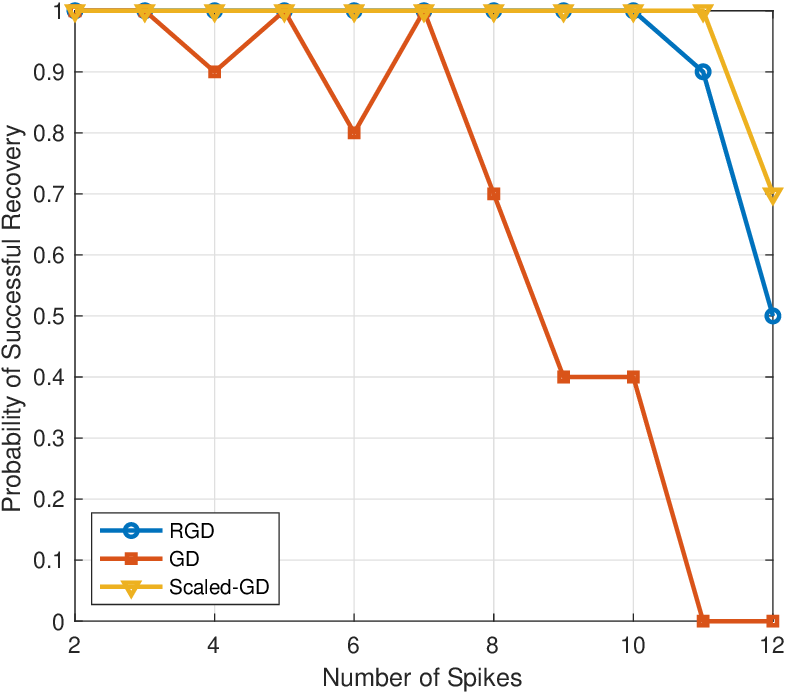}
	}
	\caption{The Empirical probability of successful recovery for RGD, GD and Scaled-GD (a) with frequency separation (b) without frequency separation.}
	\label{fig: empir prob}
\end{figure}

In the second experiment, we compare the running times of RGD, GD, and Scaled-GD. For this comparison, we fix the parameters $s=r=K=2$ and vary $n$ within the range 
$\{160:20:300\}$. We report the computational times for RGD, GD, and Scaled-GD across different values of $n$. The average computational times for all three methods are shown in Figure \ref{fig: running time}, both with and without the separation condition. The results clearly demonstrate that RGD significantly reduces running time compared to GD and Scaled-GD, particularly for larger $n$ is large, highlighting the superior efficiency of RGD.

\begin{figure}[htpb!]
	\centering
	\subfloat[]{
		\centering
		\includegraphics[width=0.23\textwidth]{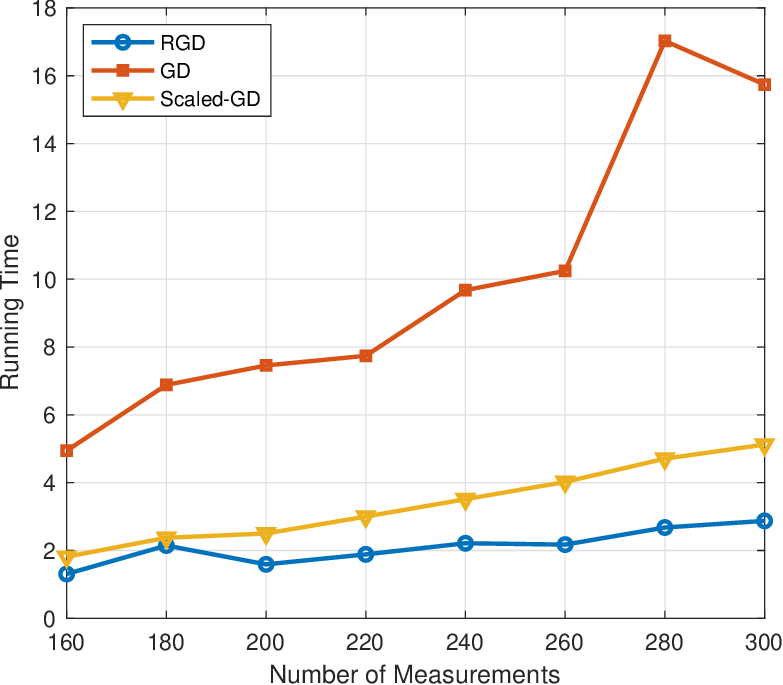}
	}
	\subfloat[]{
		\centering
		\includegraphics[width=0.23\textwidth]{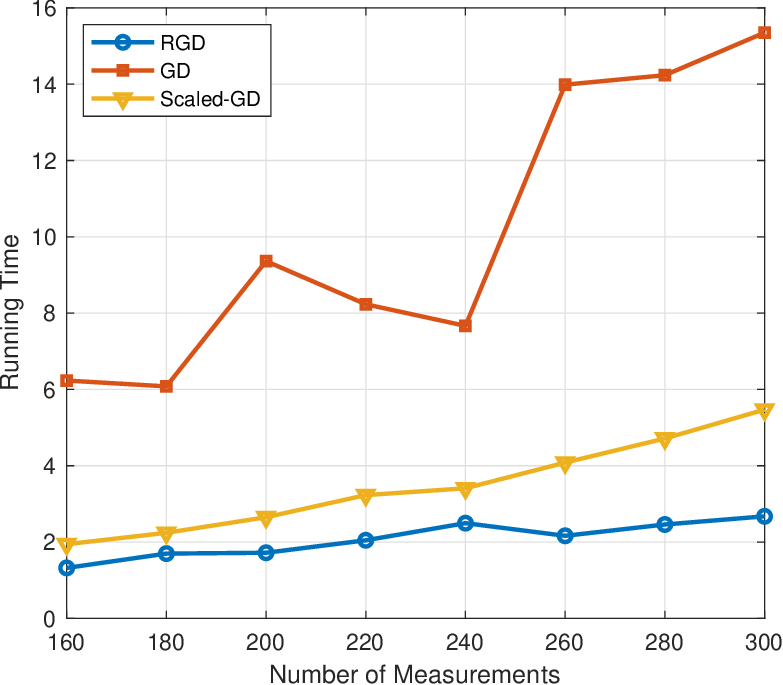}
	}
	
	\caption{The CPU running time for RGD, GD and Scaled-GD.}
	\label{fig: running time}
\end{figure}

In the third experiment, we evaluate the convergence performance of RGD in comparison to GD and Saled-GD.  For this comparison, we select $n=160, s=r=K=2$ and $n=256, s=r=4, K=2$. Fig \ref{fig: rate} shows the relative recovery error as a function of iterations for different condition numbers $\kappa = 1,5,10$. The results indicate that RGD achieves linear convergence, independent of the condition number, aligning with the predictions of our main theorem.
\begin{figure}[htpb]
	\centering
	\subfloat[]{
		\centering
		\includegraphics[width=0.23\textwidth]{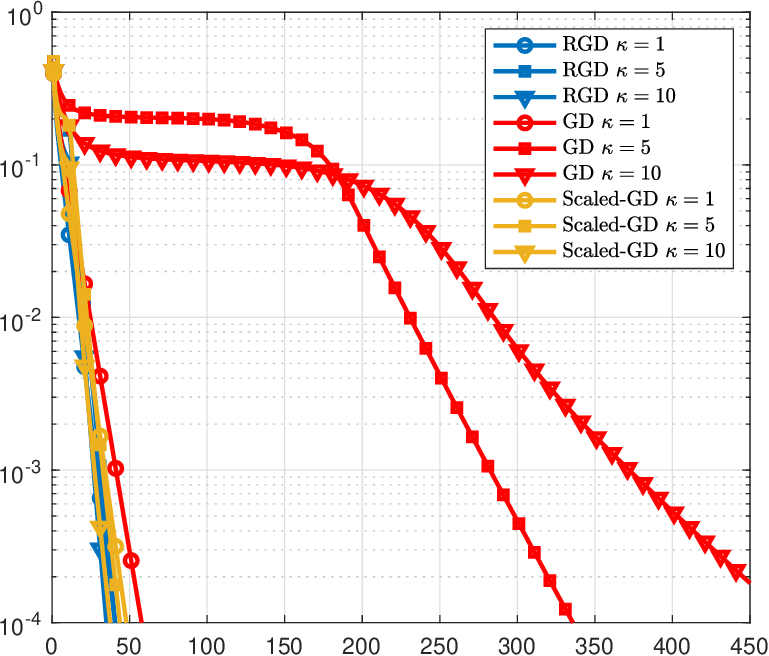}
	}%
	\subfloat[]{
		\centering
		\includegraphics[width=0.23\textwidth]{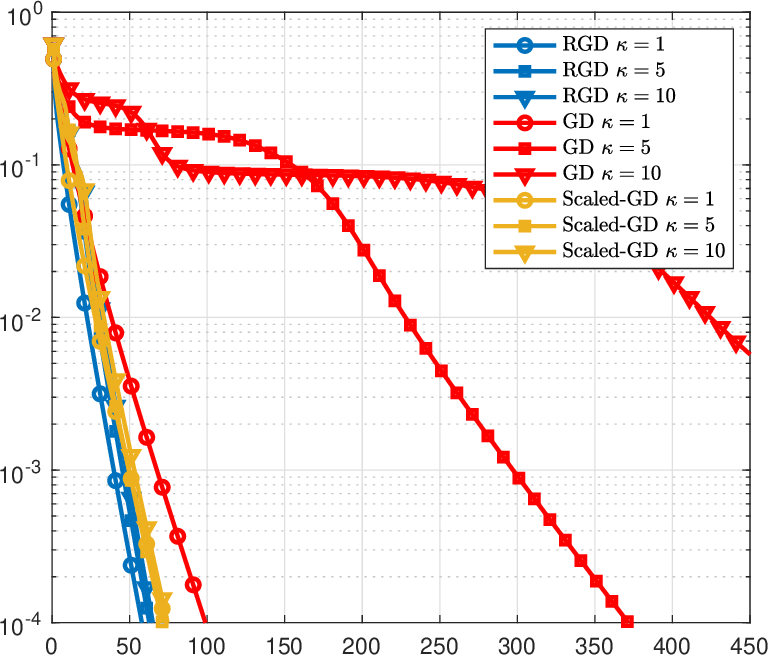}
	}
	
	\caption{Convergence rate (a) $n=160, s=r=K=2$, (b) $n=256, s=r=4, K=2$}
	\label{fig: rate}
\end{figure}


 \section{Proof of Main Result}
We will prove our main result by induction. Notice that Lemma V.7 in \cite{chen2024fast} guarantees that \eqref{main result} holds when $t=0$. Next we assume \eqref{main result} holds for the iterations $0,1,...,t$, and then prove it also holds for $t+1$.

Recall that $\mZ_t$ is a block diagonal matrix. A direct computation yields that 
 \begin{align*}
 	\fronorm{\mZ_{t+1} - \mZ_{\natural} } &\leq \fronorm{\mZ_{t+1} - \mW_t} + \fronorm{\mW_t-  \mZ_{\natural} }\\
 	&\leq 2\fronorm{\mW_t-  \mZ_{\natural} }.
 \end{align*}
 Moreover, one can express $\mW_t-  \mZ_{\natural}$ as follows:
 \begin{align*}
 	\mW_t-  \mZ_{\natural} 
 	&=  \calP_{\calS_t}\lb \mZ_{t} - \alpha_t \nabla f(\mZ_t) \rb - \mZ_{\natural}\\
 	&=(1-\alpha_t) \calP_{\calS_t}\lb \mZ_{t} - \mZ_{\natural} \rb \\
 	&\qquad+  \alpha_t  \calP_{\calS_t}\tG(\calI - \calA^\ast \calA)\tG^\ast\calP_{\calS_t}(\mZ_t - \mZ_{\natural}) \\
 	&\qquad + \alpha_t  \calP_{\calS_t}\tG\tG^\ast\lb \calI - \calP_{\calS_t}\rb(\mZ_t - \mZ_{\natural}) \\
 	&\qquad- \alpha_t  \calP_{\calS_t}\tG \calA^\ast \calA \tG^\ast\lb \calI - \calP_{\calS_t}\rb(\mZ_t - \mZ_{\natural}) \\
 	&\qquad + (\calI - \calP_{\calS_t})\lb \mZ_t - \mZ_{\natural} \rb,
 \end{align*}
 which implies that
 \begin{align*}
 	\fronorm{\mW_t-  \mZ_{\natural} }
 	\leq &(1-\alpha_t) \fronorm{ \calP_{\calS_t}\lb \mZ_{t} - \mZ_{\natural} \rb } \\
 	&+ \underbrace{ \fronorm{(\calI - \calP_{\calS_t})\lb \mZ_t - \mZ_{\natural} \rb} }_{:=I_1}\\
 	&+ \alpha_t\underbrace{ \fronorm{ \calP_{\calS_t}\tG \calA^\ast \calA \tG^\ast\lb \calI - \calP_{\calS_t}\rb(\mZ_t - \mZ_{\natural}) } }_{:=I_2}\\
 	&+\alpha_t \underbrace{ \fronorm{ \calP_{\calS_t}\tG\tG^\ast\lb \calI - \calP_{\calS_t}\rb(\mZ_t - \mZ_{\natural})  } }_{:=I_3}\\
 	&+ \alpha_t \underbrace{ \fronorm{ \calP_{\calS_t}\tG(\calI - \calA^\ast \calA)\tG^\ast\calP_{\calS_t}(\mZ_t - \mZ_{\natural}) } }_{;=I_4}.
 \end{align*}
To this end, we bound these four terms, respectively. 
\addtolength{\topmargin}{0.05in}
 \begin{itemize}
 	\item Bounding of $I_1$.
 	A direct computation yields that 
 	{\small
 	\begin{align*}
 		I_1
 		&=\sqrt{\sum_{k=1}^{K}\fronorm{(\calI - \calP_{\calS_{k,t}})\lb\mZ_{k, \natural} \rb}^2}\\
 		&\leq \sqrt{\sum_{k=1}^{K}\frac{1}{\sigma_r^2(\mZ_{k,\natural})} \fronorm{\mZ_{k,t} - \mZ_{k,\natural}}^4 }\\
 		&\leq \max_k \frac{ \fronorm{\mZ_{k,t} - \mZ_{k,\natural}} }{\sigma_r(\mZ_{k,\natural})} \cdot \sqrt{\sum_{k=1}^{K}   \fronorm{\mZ_{k,t} - \mZ_{k,\natural}}^2 }\\
 		&\leq  \frac{\epsilon}{2\sqrt{K\mu_0s (1+\epsilon)}} \cdot \fronorm{\mZ_t - \mZ_{\natural}}.
 	\end{align*}}
 	\item Bounding of $I_2$. A simple calculation yields that
 	\begin{align*}
 		I_2 &\leq \opnorm{\calP_{\calS_t}\tG \calA^\ast} \cdot \opnorm{\calA \tG^\ast}\cdot \fronorm{ \lb \calI - \calP_{\calS_t}\rb(\mZ_t - \mZ_{\natural})  } \\
 		&\leq 3\sqrt{1+\epsilon} \cdot \sqrt{K\mu_0s} \cdot  \frac{\epsilon}{2\sqrt{K\mu_0s (1+\epsilon)}} \cdot \fronorm{\mZ_t - \mZ_{\natural}}\\
 		&=\frac{3\epsilon}{2}\fronorm{\mZ_t - \mZ_{\natural}}.
 	\end{align*}
 	\item Bounding of $I_3$. One has
 	\begin{align*}
 		I_3
 		&\leq \fronorm{ \lb \calI - \calP_{\calS_t}\rb(\mZ_t - \mZ_{\natural}) }\\
 		&\leq\frac{\epsilon}{2\sqrt{K\mu_0s (1+\epsilon)}} \cdot \fronorm{\mZ_t - \mZ_{\natural}}.
 	\end{align*}
 	\item Bounding of $I_4$. Applying Lemma yields that
 	\begin{align*}
 		I_4\leq 21\epsilon \fronorm{\mZ_t - \mZ_{\natural}}
 	\end{align*}
 \end{itemize}
 Combining together, one has
 \begin{align*}
 	\fronorm{\mZ_{t+1} - \mZ_{\natural} } 
 	&\leq 2  (1-\alpha_t)  \fronorm{\mZ_t - \mZ_{\natural}}\\
 	&\quad  + 2\frac{\epsilon}{2\sqrt{K\mu_0s (1+\epsilon)}} \fronorm{\mZ_t - \mZ_{\natural}}  \\
 	&\quad + 2\alpha_t \frac{3\epsilon}{2}\fronorm{\mZ_t - \mZ_{\natural}} \\
 	&\quad + 2\alpha \frac{\epsilon}{2\sqrt{K\mu_0s(1+\epsilon)}}\fronorm{\mZ_t - \mZ_{\natural}} \\
 	&\quad + 2\alpha_t \cdot 21\epsilon\fronorm{\mZ_t - \mZ_{\natural}} \\
 	&\leq  \lb 2(1-\alpha_t) +50\alpha_t \epsilon  \rb \fronorm{\mZ_t - \mZ_{\natural}}\\
 	&\leq \frac{1}{2}\fronorm{\mZ_t - \mZ_{\natural}},
 \end{align*}
 where the last line is due to $1 \geq \alpha_t \geq \frac{7}{8}$. Thus we complete the proof.
 
 \subsection{Useful Lemmas}
 \begin{lemma}\cite[Lemma VII.3]{chen2024fast}
 	\label{lemma: T-RIP}
 	Suppose $n\geq C_{\gamma} \epsilon^{-2}K^2\mu_0 s\mu_1r\log(sn)$. Then with probability at least $1-(sn)^{-\gamma+1}$, there holds the following inequality
 	\begin{align*}
 		\opnorm{\calP_T\tG(\calA^\ast \calA - \calI)\tG^\ast \calP_T} \leq \epsilon .
 	\end{align*}
 \end{lemma}

 \begin{lemma}
 	Suppose that then
 	\begin{align*}
 		\max_k \frac{\fronorm{\mZ_{k,t}- \mZ_{k, \natural}} }{\sigma_r(\mZ_{k, \natural})}\leq \frac{\epsilon}{\sqrt{K\mu_0s(1+\epsilon)}}
 	\end{align*}
 	Conditioned on Lemma \ref{lemma: T-RIP}, one has
 	\begin{align*}
 		\opnorm{\calA\tG^\ast \calP_{\calS_t}} &\leq 3\sqrt{1+\epsilon},\\
 		\opnorm{ \calP_{\calS_t}\tG(\calI - \calA^\ast \calA)\tG^\ast\calP_{\calS_t}} &\leq 21\epsilon.
 	\end{align*}
 \end{lemma}
 \begin{proof}
 	For any block diagonal matrix $\mY$ such that $\fronorm{\mY}=1$, one has
 	\begin{align*}
 	&\fronorm{\calA\tG^\ast \calP_{\calS}(\mY)}^2 \\
 		&= \la \mY, \lb \calP_{\calS}\tG\calA^\ast\calA\tG^\ast \calP_{\calS} - \calP_{\calS}\rb(\mY)\ra + \fronorm{\calP_{\calS}(\mY)}^2\\
 		&\leq \opnorm{ \calP_{\calS}\tG\calA^\ast\calA\tG^\ast \calP_{\calS} - \calP_{\calS} }\cdot \fronorm{\mY}^2 + 1\\
 		&\leq 1+\epsilon.
 	\end{align*}
 	Furthermore, one has
 	\begin{align*}
 		\fronorm{ \lb \calP_{\calS_t} - \calP_{\calS} \rb(\mY) }
 		&\leq \sqrt{ \sum_{k=1}^{K} \opnorm{ \calP_{\calS_{k,t}} - \calP_{\calS_k} }^2 \cdot \fronorm{\mY_k}^2 } \\
 		&\leq \sqrt{ \sum_{k=1}^{K} \lb \frac{2\fronorm{\mZ_{k,t} - \mZ_{k, \natural}}}{\sigma_r(\mZ_{k, \natural})}\rb^2 \cdot \fronorm{\mY_k}^2 } \\
 		&\leq \max_k \frac{2\fronorm{\mZ_{k,t} - \mZ_{k, \natural}}}{\sigma_r(\mZ_{k, \natural})} \\
 		&\leq \frac{2\epsilon}{\sqrt{K\mu_0s (1+\epsilon)}},
 	\end{align*}
 	which implies that $\opnorm{ \calP_{\calS_t} - \calP_{\calS} } \leq \frac{2\epsilon}{\sqrt{K\mu_0s (1+\epsilon)}}$.
 	A simple computation yields that 
 	\begin{align*}
 		\fronorm{\calA\tG^\ast \calP_{\calS_t}(\mY)}&\leq \fronorm{\calA\tG^\ast \lb \calP_{\calS_t} - \calP_{\calS} \rb(\mY)} \\
 		&\quad + \fronorm{\calA\tG^\ast \lb\calP_{\calS} \rb(\mY)} \\
 		&\leq \opnorm{\calA\tG^\ast} \cdot \fronorm{ \lb \calP_{\calS_t} - \calP_{\calS} \rb(\mY) } + \sqrt{1+\epsilon}\\
 		&\leq \sqrt{K\mu_0s} \cdot \frac{2\epsilon}{\sqrt{K\mu_0s (1+\epsilon)}}.  + \sqrt{1+\epsilon}\\
 		&\leq 3\sqrt{1+\epsilon}.
 	\end{align*}
 	Finally, one has
 	\begin{align*}
 		&\calP_{\calS_t}\tG(\calI - \calA^\ast \calA)\tG^\ast\calP_{\calS_t}\\
 		=&\lb \calP_{\calS_t} - \calP_{\calS}\rb \tG\tG^\ast \calP_{\calS_t}  +  \calP_{\calS}\tG\tG^\ast \lb \calP_{\calS_t}  -  \calP_{\calS}\rb \\
 		&\qquad +  \calP_{\calS}\tG\tG^\ast  \calP_{\calS} - \calP_{\calS}\tG\calA^\ast \calA\tG^\ast  \calP_{\calS}\\
 		&\qquad +\lb \calP_{\calS} - \calP_{\calS_t} \rb\tG\calA^\ast \calA\tG^\ast  \calP_{\calS}\\ &\qquad + \calP_{\calS_t} \tG\calA^\ast \calA\tG^\ast  \lb \calP_{\calS}- \calP_{\calS_t} \rb,
 	\end{align*}
 	which implies that
 	\begin{align*}
 		&\opnorm{ \calP_{\calS_t}\tG(\calI - \calA^\ast \calA)\tG^\ast\calP_{\calS_t} } \\
 		&\leq \opnorm{ \lb \calP_{\calS_t} - \calP_{\calS}\rb \tG\tG^\ast \calP_{\calS_t} } + \opnorm{ \calP_{\calS}\tG\tG^\ast \lb \calP_{\calS_t}  -  \calP_{\calS}\rb } \\
 		&\qquad + \opnorm{ \calP_{\calS}\tG\tG^\ast  \calP_{\calS} - \calP_{\calS}\tG\calA^\ast \calA\tG^\ast  \calP_{\calS} }\\
 		&\qquad +\opnorm{\lb \calP_{\calS} - \calP_{\calS_t} \rb\tG\calA^\ast \calA\tG^\ast  \calP_{\calS} } \\
 		&\qquad + \opnorm{ \calP_{\calS_t} \tG\calA^\ast \calA\tG^\ast  \lb \calP_{\calS}- \calP_{\calS_t} \rb}\\
 		&\leq 2\opnorm{  \calP_{\calS_t} - \calP_{\calS} }  + \epsilon \\
 		&\qquad +2 \opnorm{\calA}\cdot \opnorm{  \calP_{\calS_t} - \calP_{\calS} }   \cdot \lb  \opnorm{\calA\tG^\ast \calP_{\calS} } + \opnorm{\calA\tG^\ast \calP_{\calS_t}} \rb\\
 		&\leq \epsilon + 2\cdot  \frac{2\epsilon}{\sqrt{K\mu_0s (1+\epsilon)}} \\
 		&\qquad + 2\sqrt{K\mu_0s} \cdot \frac{2\epsilon}{\sqrt{K\mu_0s (1+\epsilon)}}  \cdot \lb \sqrt{1+\epsilon} + 3\sqrt{1+\epsilon}\rb\\
 		&\leq \epsilon + 4\epsilon + 16\epsilon \\
 		&\leq 21\epsilon.
 	\end{align*}
 \end{proof}

 \section{Conclusion}
In this work, we investigate the problem of simultaneous blind super-resolution and demixing in ISAC, formulating it as a low-rank matrix demixing problem. We propose an RGD method to solve this problem and establish its sample complexity, along with a linear convergence guarantee. Notably, we demonstrate that the convergence rate is independent of the condition number of the target matrices. The empirical effectiveness of our algorithm is validated through extensive numerical experiments.

\bibliographystyle{IEEEbib}

\bibliography{refBlind}

\end{document}